%% file: Draft1.tex
\documentclass[conference,compsoc]{IEEEtran}
\usepackage{amsmath}
\usepackage{amsthm}
\usepackage{amssymb}
\usepackage{verbatim}
\usepackage{graphicx}
\usepackage{geometry}
\usepackage{tikz}
\usepackage{pgfplots}
\usetikzlibrary {positioning}
\usetikzlibrary{shapes,snakes}
\usetikzlibrary{arrows,calc}
\usepackage{relsize}
\usetikzlibrary{patterns}
 \usepackage{algorithm}
\usepackage{algpseudocode}

\theoremstyle{plain}
\newtheorem{lemma}{Lemma}

\newtheorem{theorem}{Theorem}

\newtheorem{definition}{Definition}

\theoremstyle{definition}

\floatname{algorithm}{Algorithm}

\allowdisplaybreaks[4]

\input{input.tex}

\begin{document}

\newlength{\figurewidth}\setlength{\figurewidth}{0.6\columnwidth}



\title{\fontsize{23}{23}\selectfont Energy-Delay-Distortion Problem}

\newcounter{one}
\setcounter{one}{1}
\newcounter{two}
\setcounter{two}{2}

\addtolength{\floatsep}{-\baselineskip}
\addtolength{\dblfloatsep}{-\baselineskip}
\addtolength{\textfloatsep}{-\baselineskip}
\addtolength{\dbltextfloatsep}{-\baselineskip}
\addtolength{\abovedisplayskip}{-1ex}
\addtolength{\belowdisplayskip}{-1ex}
\addtolength{\abovedisplayshortskip}{-1ex}
\addtolength{\belowdisplayshortskip}{-1ex}

%


\author{\IEEEauthorblockN{Rahul Vaze}
\IEEEauthorblockA{
Tata Institute of Fundamental Research\\
 Homi Bhabha Road, Mumbai 400005\\
Email: vaze@tcs.tifr.res.in}
\\
\IEEEauthorblockN{Akshat Choube}
\IEEEauthorblockA{Indian Institute of Technology\\
Palakkad, India\\
Email:111501031@smail.iitpkd.ac.in }
\and
\IEEEauthorblockN{Shreyas Chaudhari}
\IEEEauthorblockA{Indian Institute of Technology,\\
Madras, India,
\\
Email: shreyaschaudhari@gmail.com\\
}\\
 
\IEEEauthorblockN{Nitin Aggarwal}
\IEEEauthorblockA{Indian Institute of Technology,\\
Roorkee, India,
\\
Email: nitinagg1235@gmail.com\\
}

}

%
\maketitle
\begin{abstract}
An energy-limited source trying to transmit multiple packets to a destination with possibly different sizes is considered. With limited energy, the source cannot potentially transmit all bits of all packets. In addition, there is a delay cost associated with each packet. Thus, the source has to choose, how many bits to transmit for each packet, and the order in which to transmit these bits, to minimize the cost of distortion (introduced by transmitting lower number of bits) and queueing plus transmission delay, across all packets. Assuming an exponential metric for distortion loss and linear delay cost, we show that the optimization problem is jointly convex. Hence, the problem can be exactly solved using convex solvers, however, because of the complicated expression derived from the KKT conditions, no closed form solution can be found even with the simplest cost function choice made in the paper, also the optimal order in which packets should be transmitted needs to be found via brute force. To facilitate a more structured solution, a discretized version of the problem is also considered, where time and energy are divided in discrete amounts. In any time slot (fixed length), bits belonging to any one packet can be transmitted, while any discrete number of energy quanta can be used in any slot corresponding to any one packet, such that the total energy constraint is satisfied. The discretized problem is a special case of a multi-partitioning problem, where each packet's utility is super-modular and the proposed greedy solution is shown to incur cost that is at most $2$-times of the optimal cost.

\end{abstract}
\section{Introduction} \label{sec:intro}
Rate-distortion problem is a classical problem, where the objective is to find minimum transmission rate to support a given distortion constraint under a specific distortion metric. Typically, the problem is considered for single source-destination pair, with average power constraints, and optimal results on the rate-distortion problem are derived when infinitely large blocklengths are allowed \cite{Cover2004}. Rate-distortion with finite blocklengths has been considered in \cite{gupta2008rate} and \cite{kostina2012fixed}.

Real-time communication, communication under quality-of-service (QoS) constraint, energy harvesting communication etc., only allows for short delays with limited energy (not necessarily average power constraint). To address the distortion problem under this more practical regime,   
 we consider in this paper that at the beginning of communication, there are $n$ packets available with the source, each with possibly distinct sizes ($B_i$ bits). The total amount of energy available with the source is limited, and hence the source can only transmit a fraction (${\hat B}_i$ bits) of each packet that introduces/forces a distortion. Each packet has a cost that consists of two components: distortion and delay. For each packet, distortion measures the gap in $B_i$ and ${\hat B}_i$ bits, and is naturally a decreasing function of ${\hat B}_i$, while delay consists of queuing delay (transmission time of other packets before it) plus its own transmission delay. 

We choose a natural distortion function $2^{B_i - {\hat B}_i}$ for each packet, that has the diminishing returns property, i.e., the rate of decrease in cost decreases with increasing ${\hat B}_i$ and is convex. 
The overall cost is the sum of the cost of all packets, and the problem is to find optimal ${\hat B}_i$, transmission time $t_i$ for each packets, and the order in which these $n$ packets are sent to minimize the overall cost. We assume a lossless communication model, where if energy $e$ is used for time $t$, $b$ bits can be send using the Shannon formula $b = t \log_2(1+\frac{e}{t})$.

Our choice of distortion function has also been used in a related problem \cite{orhan2015source}, where the cost function is just the sum of the distortion for each packet, but does not include any delay cost. In \cite{orhan2015source}, without the delay cost, optimal closed form solutions have been obtained. On the other hand, the problem where only delay cost is counted and the distortion cost is neglected has been considered in \cite{yang2008delay}. Other variants of rate-distortion problems under energy constraints without delay cost can be found in \cite{zordan2014performance, biason2017energy, bhat2016distortion, arafa2017near}.

The problem considered in this paper is also related to the scheduling problem considered in \cite{hou2009theory}, where at the beginning of each frame, $n$ packets of equal size arrive at the source, and either they are transmitted successfully by the end of the frame or dropped completely. A lossy model for communication is used in \cite{hou2009theory}, where in each slot of the frame, a packet transmitted is successful with probability $p$ or erased otherwise, and the decision variable is to decide which packet to transmit in any slot among the ones that have not been transmitted successfully by then, to maximize throughput.

An online version of the considered problem, which is the part of ongoing work, considers that $n$ packets (each with possibly distinct sizes, $B_i$ bits) arrive at the source at distinct times and under the cost function described above, the problem is to find how many bits to send for each packet, when to begin its transmission, and for how long to transmit its bits. A moment's thought will reflect that the online version is a general case of the age of information problem \cite{kaul2012real}, where $B_i=1$ and ${\hat B}_i \in \{0,1\}$, i.e., the problem is to minimize the delay between the time at which the packet with one bit arrives at the source and the time at which the receiver knows about it, if at all. 

Thus, the problem formulation introduced in this paper is quite general, and addresses two important problems : rate-distortion problem with finite delays (where rate constraint is an artefact of limited energy), and the generalized age of information problem, where the generalization includes the dependence of the identity of packets and their sizes on the cost function.

Our contributions are as follows:
\begin{itemize} 
\item We show that the considered energy-delay-distortion problem is jointly convex for a fixed order of transmission of packets, and hence can be solved using any convex solver. Unfortunately, however, even for the simplest choice of reasonable cost function, the KKT conditions (though sufficient for optimality) cannot be used to find a structured solution, and the optimal order of packet transmission. This is in contrast to \cite{orhan2015source}, where closed form solution is found when the delay cost is not included. Thus, including the delay cost, not only makes the problem more practically relevant, but is also fundamentally different analytically than \cite{orhan2015source}. Moreover, the joint convexity of the problem, does not reveal anything about the optimal order in which packets should be transmitted, and the intuitive choice of transmitting shorter packets first is not provable, though seen via simulations.
\item To get a structured solution that does not require a brute force search over all possible orders of packet transmissions, we also consider a discretized version of the problem. In the discretized version,  time is divided into discrete slots, and any one slot can be used only to transmit bits belonging to the same packet. In addition, we also discretize the energy into small units of $e$ each, that is the least amount of energy that will be used in one slot. Thus, the equivalent problem is to find an allocation of {\it resource blocks} (rectangles of height (energy) $e$ and width (slot time $\ell$)) to packets, under the total energy constraint, such that the objective function is minimized. 
\item This discretized version is a discrete optimization problem, which in general is harder to solve compared to a continuous (and convex problem in this case) one. The structure of the problem, however, comes to the rescue by noting the fact that the discrete problem is a special case of the multi-partitioning problem, where the objective is to partition a given set of resources among multiple agents to minimize an overall objective function. For the discretized problem, we show that a greedy algorithm achieves at most twice the cost of the optimal solution, via exploiting the super-modularity of the cost function for each of the packets. Thus, the discretized model allows the use of  a simple structured solution that is guaranteed to be close to the optimal. 
\end{itemize}

\section{Problem Formulation} 
Consider a source that has $n$ packets with $B_i, i=1,\dots, n$ bits each, that it wants to communicate to its destination. 
The total energy available with the source is $E$, that can be used to transmit any bits of the $n$ packets. Finite $E$ limits the number of bits that can be sent from the source to its destination, and thus the source has to judiciously choose how many bits of each packet can be sent, and the order in which the packets should be sent since that also impacts the QoS.

To make this precise, let the source send ${\hat B}_i$ out of $B_i$ bits of packet $n$ using energy $E_i$ and time $t_i$. The actual method to compress $B_i$ bits to ${\hat B}_i$ bits is out of scope of this paper, and can be found in quantization literature.
Let the $i^{th}$ packet be sent at the $\pi(i)^{th}$ location, then the cost for packet $i$ is defined as 
$$U_i = 2^{B_i-{\hat B}_i} + \sum_{k\in \pi(k) \le \pi(i)} t_k + t_i,$$
where the second term is the queuing delay.
The overall cost of the source is 
\begin{equation}\label{eq:util}
U = \sum_{i=1}^n U_i.
\end{equation}
Then the optimization problem is 
\begin{equation}\label{eq:prob}
\min_{E_i, t_i, \sum_{i=1}^n E_i \le E}U.
\end{equation}
The choice of $U_i$ is motivated by the fact that any natural cost function has a diminishing returns property such that its incremental decrease reduces as ${\hat B}_i$ increases. The delay component counts the delay of packet $i$ as well as the queuing delay that it experiences because of transmission of packets transmitted before it.

Problem \ref{eq:prob} has connections with the rate-distortion theory in finite time and energy, which to the best of our knowledge is unsolved. To be specific, the first term of $U_i$, $2^{B_i-{\hat B}_i}$ measures the distortion for packet $i$, and the rate restriction follows because of finite energy $E$ and the presence of other packets. The linear delay term weights the rate at which packets are being delivered to the destination. One can keep $U_i$ a general function of $B_i, {\hat B}_i$, however, the specific choice made here is quite natural, that has diminishing returns property, and is convex, without making the problem trivial.

We use the Shannon rate formula to relate the ${\hat B}_i$, $E_i$ and $t_i$ for packet $i$, that is given by 
$${\hat B}_i = t_i \log \left(1+\frac{E_i}{t_i}\right).$$ Using this rate formula, we can write Problem \ref{eq:prob}, as a function of $E_i$ or $t_i$ alone.
Even under this 'simple' rate formula, Problem \ref{eq:prob}
remains challenging, where even figuring out the optimal order in which packets should be sent cannot be solved in closed form. 
\section{Optimal Solution For Problem \ref{eq:prob}}
In the following, we first establish that Problem \ref{eq:prob} is jointly convex problem under the Shannon-rate formula. Establishing that Problem \ref{eq:prob} is convex in both $E_i$ or $t_i$ individually is rather easy.
\begin{theorem}\label{thm:jointconv} For a fixed order of packet transmission $\pi$, Problem \ref{eq:prob} is jointly convex problem in $E_i$ and $t_i$.
\end{theorem}
Proof can be found in Appendix \ref{app:jointconv}.
Using the joint convexity, Problem \ref{eq:prob} can be solved efficiently by any of the convex solvers. Importantly, however, finding the optimal order of packet transmission $\pi$ remains a challenge. One bottleneck in doing so is that the KKT conditions for this problem, which will be sufficient because of Theorem \ref{thm:jointconv}, lead to exponential functions in the variables of interest that cannot be solved in closed form. 
The KKT conditions can be obtained directly by taking a derivative of \eqref{eq:util}, however, we  do not describe them here since they are unwieldy.  Thus, the KKT conditions do not allow any analytical tractability, and no structure for the optimal solution can be extracted from the KKT conditions. Thus, one has to rely on the convex solvers to solve this problem for a fixed $\pi$, and optimize over $\pi$ thereafter. In the next section, we present a structural solution to the problem that obviates the need for optimizing over $\pi$, which essentially has exponential complexity.

\section{Discretized Variant of Problem \ref{eq:prob}}
In this section, we work towards finding a more structured solution for a discretized variant of Problem \ref{eq:prob} for which we can find theoretical guarantees on the performance. To facilitate this, we consider a discretized version of Problem \ref{eq:prob}, where time is divided in discrete slots of short fixed length $\ell$. In each slot, bits from at most one packet can be sent, however, bits from the same packet can be sent in multiple non-contiguous slots. 
Let set $\cP_i$ be the set of slots assigned to packet $i$. We also discretize the energy into small units of $e$ each, that is the least amount of energy that will be used in one slot. We define a {\it resource block}  as a rectangle of height (energy) $e$ and width (slot time $\ell$). For slot $j$, the number of resource blocks is defined as $R_j$. Since any one slot is reserved for bits from any one packet, $e R_j$ is the amount of energy used for transmission of bits for packet $i$ if $j \in \cP_i$. 

For slots $j \in \cP_i$, the cumulative bits sent using resource blocks $R_j$ is ${\hat B}_i$, where  
$${\hat B}_i =  \ell \log \left(1+\frac{e R_j}{\ell}\right).$$  Then the cost $D_i$ for packet $i$ is 
\begin{equation}\label{eq:Di}
D_i = 2^{(B_i-\sum_{j \in \cP_i}{\hat B}_i)} + \ell \ i_{max},
\end{equation}
where $i_{max} = \max \{j : j \in \cP_i\}$ is the last slot where any bits of packet $i$ are sent.

The total energy consumption under this setup is $\sum_{j=1}^J e R_j$, where $J$ is the total number of slots used for transmission.
Then the optimization problem is 
\begin{equation}\label{eq:discprob}
\min_{\cP_i, R_j, \sum_{j=1}^J e R_j \le E} \sum_{i=1}^n D_i.
\end{equation}
Problem \ref{eq:discprob} is a discrete optimization problem, since $\cP_i$ and $R_j$ are discrete sets and only one packet can be assigned to any one slot. Thus, it is not evident that Problem \ref{eq:discprob} is any easier than Problem \ref{eq:prob}. To understand how to efficiently solve Problem \ref{eq:discprob}, we need the following preliminaries.

\begin{definition}\label{definition:submod}
Let $V$ be a finite set, and let $2^V$ be the power set of $V$. A real-valued set function $f:2^V\to\mathbb{R}$ is said to be {\it monotone} if $f(S) \ge f(T)$ for $S \subseteq T \subseteq V$,  and {\it super-modular} if 
\begin{equation}
  \label{eq:22}
f(S)+f(T)\le f(S\cap T)+f(S\cup T),\ \forall S,T\in 2^V.   
\end{equation}
An equivalent definition of super-modularity is 
\begin{equation}\label{eq:usefuldefnsubmod0}
f(S \cup \{i\}) + f(S \cup \{j\}) \le f(S) + f(S \cup \{i, j\})
\end{equation}
for every $S\subseteq V$ and every $i,j\in V\setminus S$ with $i\ne j$. 
\end{definition}

Let $\bS$ be a finite set and let $f_i$, $1\le i\le k$, be set
functions from $2^\bS$ to the real numbers. The multi-partitioning
problem is defined as follows.
\begin{definition}{\bf Multi-partitioning problem:}\label{defn:multipart}
Partition a given (resource) set $\bS$  into $k$ subsets $S_1, \dots, S_k, S_i \cap S_j = \phi, \cup_{i=1}^k S_i = \bS$ such that $\sum_{i=1}^k f_i(S_i)$ is minimized.
\end{definition}

\begin{theorem}[\cite{Nemhauser1978,Lehmann2006combinatorial}]\label{theorem:Nemhauser}
If all functions $f_i$ in the multi-partitioning problem are
  non-negative, monotone, and super-modular, then the $\greedy_1$
  algorithm outputs a partition whose cost is at most 
  twice that of the optimal partition.
\end{theorem}

Problem \ref{eq:discprob} is a multi-partitioning problem, where the resource set $\bS$ is the set of resource blocks $R_j$ for $\cP_i$ that needs to be partitioned among the $n$ packets, such that the total energy constraint is satisfied $\sum_{j=1}^J e R_j \le E$, where slot $j \in P_i$ for some $i$ (packet). 
One major difference is in the definition of resource blocks that are {\it dynamic} for Problem \ref{eq:discprob} rather than being fixed ahead of time. To be clear, given a set of existing resource blocks, 
$R_j, j \in \cP_i$, $R_j \rightarrow R_j+1$ is allowed only when if additional bits from packet $i$ are sent using extra energy $e$ in that slot. Moreover, for any $R_{j_1}, R_{j_2}, j_1, j_2 \in \cP_i$, 
a new resource block is added to $R_{j_1}$ or $R_{j_2}$, depending on which one reduces the cost for packet $i$ more. Only fixed constraint is that the total number of resource blocks is at most $E/e$, i.e., $\sum_{j=1}^J e R_j \le E$.

\begin{figure}
\begin{tabular}{r l}
\hline
& $\greedy_1$ algorithm\\
\hline
1	&	Initialize  $S_j = \Phi$, $j =1,\ldots, k$, $i=1$. \\
2	&	Find $j^\star =\arg\min_{j}f_{j}(S_j\cup\{i\})$.	\\
3 &	Update $S_{j^\star} = S_{j^\star} \cup \{i\}$.\\
4	&	Set $i=i+1$, Stop if $i>|\bS|$, \\
      & otherwise go to step 2.\\
5	&	Return $S_j$, $j =1,\dots, k$.\\
\hline
\end{tabular} 
\caption{Greedy allocation algorithm for the multi-partitioning
  problem.}
\label{fig:greedyalgo}
\end{figure}

To solve Problem \ref{eq:discprob}, consider the $\greedy$ algorithm in Fig. \ref{fig:newgreedy}, which 
allocates a new resource block to the packet that reduces the incremental cost the most. The novelty in this greedy algorithm is that both the packet index set $\cP_i$ (which slot to assign for packet $i$) and which packet to assign to a new resource block (if created at a previously un-allotted slot) is being found out greedily, since given the existing set of resource blocks $R_j, j \in \cP_i$, a new resource block can be created at any of the existing slots where $R_j > 0$ i.e. $j \in \cP_i$ for some $i=1,\dots, n$ or at a new slot where $R_j =0, j \notin \cP_i$ for any $i = 1,\dots, n$.

We illustrate the functioning of the $\greedy$  algorithm for solving Problem \ref{eq:discprob} in Fig. \ref{fig:greedy}. In the considered iteration of the $\greedy$ algorithm, the distinct colored (other than cyan) rectangles of Fig. \ref{fig:greedy} are resource blocks that have been already assigned to different packets, where the same color represents blocks that are assigned to the same packet. 
The new resource block that is to be assigned is among the candidate resource blocks (denoted A) that could either be allocated to one of the slots ($j^\star$) that is already occupied by some packet or a completely new slot, depending on which choice makes the largest decrease in cost, given the earlier allocation. In case, a resource block is assigned to a previously empty slot, the algorithm also describes which packet ($i^\star$) should be assigned to that block.

\begin{theorem}\label{theorem:newgreedy}

If all functions $D_i$ in Problem \ref{eq:discprob} are
  non-negative, monotone, and super-modular, then the $\greedy$ algorithm \ref{fig:newgreedy} outputs a partition (allocation of resource blocks) within a factor of $2$ of the optimal partition.
\end{theorem}
Proof is similar to Theorem \ref{theorem:Nemhauser}.
\begin{figure*}
\begin{tabular}{r l}
\hline
& $\greedy$ algorithm\\
\hline
1	&	Initialize $\cP_i = \phi$ $R_j = 0, R_j \in \cP_i$, $\forall \ i$. \\
2.  & Let $\text{last} = \max\{j: j \in \cP_i\}$ for some $i$  \%Last slot that has been assigned to any packet \\
3. &{\bf For} slot $j=1:\text{last}$ \\ 
4. & If{ energy is not exhausted : $\sum_{j=1}^{\text{last}} e R_j \le E$}.\\
5. &Find the slot $j$ (or packet $i : j \in \cP_i$) that benefits the user $i$ most by allocating  \\
 & a new resource block with $R_j=R_j+1$, i.e., \\
 &$j^\star=\arg\min_{j=1, \dots, \text{last}, j \in \cP_i} D_i(R_{-j} \cup R_j+1)$ where $R_{-j} = \cup_{k \in \cP_i} R_k \backslash R_{j}$\\%
\\
6.  & Find the packet $i, i=1,\dots, n$ that benefits most by allocating first resource block at slot $\text{last}+1$, i.e.,  \\
& $i^\star=\arg\min_{i=1, \dots, n, j \in \cP_i} D_i(R_j \cup 1_{\text{last}+1})$ \\
\\
7. & Make $R_{j^\star} = R_{j^\star} +1$ if  $\min_{j=1, \dots, \text{last}, j \in \cP_i} D_i(R_{-j} \cup R_j+1) < \min_{i=1, \dots, n, R_j \in \cP_i} D_i(R_j \cup 1_{\text{last}+1})$\\
& Otherwise create a new slot and assign it to $i^\star$, 
i.e., slot $\text{last}+1 \in \cP_{i^\star}$ and $R_{\text{last}+1} = 1$
 \\
8.    &            Update $\text{last} = \max\{j: j \in \cP_i\}$, go to step 3\\
 & {\bf End}  \\
9.	&	Return $R_j$, $j \in \cP_i$.\\
\hline
\end{tabular} 
\caption{Greedy allocation algorithm for Problem \ref{eq:discprob}.}
\label{fig:newgreedy}
\end{figure*}


To make use of Theorem \ref{theorem:newgreedy}, we now show that the cost functions $D_i$ are non-negative, monotone and super-modular as follows.
\begin{lemma} The cost function $D_i$ is non-negative, monotone and super-modular
\end{lemma}
\begin{proof}
The non-negativity of $D_i$ is obvious, since the first term is an exponential function, while the second term is linear. To show monotonicity, we need to show that $D_i(\cup_{j \in \cP_i} R_{j} \cup \{r\}) \le D_i(\cup_{j \in \cP_i} R_{j})$ for any packet $i$, where $r$ is a new resource block that is not part of  $\cup_{j \in \cP_i} R_{j}$. By the definition of the resource block as described earlier, a new resource block corresponding to  packet $i$ is either added to the slots that are already allotted to that packet, i.e., $\cP_i$ or to an un-allotted  slot depending on which ever one gives larger decrease in cost. Adding a new resource block to the existing time slot $\cP_i$ clearly increases the number of bits ${\hat B}_i$ for packet $i$ while not increasing the delay, thereby decreasing the cost $D_i$. Thus, adding a new resource block to either $\cP_i$ or an un-allotted slot cannot increase the cost $D_i$, thus proving monotonicity.
The super-modularity of $D_i$ is also easy to see since function $2^{B_i - {\hat B}_i}$ is convex and the delay term is linear. 
\end{proof}

Thus, we have the following Theorem for Problem \ref{eq:discprob}. 

\begin{theorem}\label{theorem:discprob}
The resource block allocation output by the $\greedy$ algorithm \ref{fig:newgreedy} for  Problem \ref{eq:discprob} has cost that is at most $2$ times the optimal cost. 
\end{theorem}

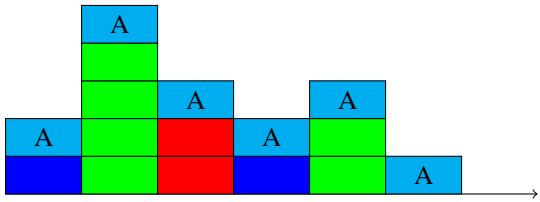
\begin{figure}\label{fig:greedy}
\begin{centering}
\begin{tikzpicture}[rahulstyle/.style={rectangle, draw}]
\draw [->] (0,0) -- (7,0);
\draw [fill=blue,blue, draw=black](0,0) rectangle (1,.5);
\draw [fill=cyan,cyan, draw=black](0,.5) rectangle (1,1);

\draw [fill=green,green, draw=black](1,0) rectangle (2,.5);
\draw [fill=green,green, draw=black](1,.5) rectangle (2,1);
\draw [fill=green,green, draw=black](1,1) rectangle (2,1.5);
\draw [fill=green,green, draw=black](1,1.5) rectangle (2,2);
\draw [fill=cyan,cyan, draw=black](1,2) rectangle (2,2.5);

\draw [fill=red,red, draw=black](2,0) rectangle (3,.5);
\draw [fill=red,red, draw=black](2,.5) rectangle (3,1);
\draw [fill=cyan,cyan, draw=black](2,1) rectangle (3,1.5);

\draw [fill=blue,blue, draw=black](3,0) rectangle (4,.5);
\draw [fill=cyan,cyan, draw=black](3,.5) rectangle (4,1);

\draw [fill=green,green, draw=black](5,0) rectangle (4,.5);
\draw [fill=green,green, draw=black](5,.5) rectangle (4,1);
\draw [fill=cyan,cyan, draw=black](5,1) rectangle (4,1.5);

\draw [fill=cyan,cyan, draw=black](6,0) rectangle (5,.5);

%
%
%
%
%
%
%

\node [right,text width=3cm] at (0.25,.75) {A};
\node [right,text width=3cm] at (1.25,2.25) {A};
\node [right,text width=3cm] at (2.25,1.25) {A};
\node [right,text width=3cm] at (3.25,.75) {A};
\node [right,text width=3cm] at (4.25,1.25) {A};
\node [right,text width=3cm] at (5.25,.25) {A};



\end{tikzpicture}
\caption{Pictorial description of the $\greedy$ algorithm, where A are the candidates slots available for allocation in any iteration.}
\end{centering}
\end{figure}

\section{Simulation Results}
In this section, we present some simulation results for the optimal solution output by convex solvers for Problem \ref{eq:prob}. In Fig. \ref{fig:monotone}, we consider two packets, and total energy $E=50$ Joules. We plot two curves for the overall cost $U$ in Fig. \ref{fig:monotone}, where in each, the size of $B_1$ or $B_2$ is kept fixed, while the other is varied. For both the curves, we fix the order of transmission as packet $1$ first followed by packet $2$. The available energy is insufficient to transmit 
$B_1=15$ bits and $B_2 = 20$ bits completely, and the optimal algorithm sends 
${\hat B}_1 = 13.667$  and ${\hat B}_2= 19.1396$, bits respectively. We consider another two packet setting in Fig. \ref{fig:shortex} with lower energy $E=20$ Joules, where again the energy is insufficient to transmit $B_1=12$ bits and $B_2 = 20$ bits and the optimal algorithm sends 
${\hat B}_1 = 7.663$ bits and ${\hat B}_2= 15.8431 $ bits. 
The inference to draw from Figs. \ref{fig:monotone} and \ref{fig:shortex}, that for both the curves in both Figs., sending the shorter packet first is optimal in terms of minimizing the cost, which is intuitive given the cost function, however, is difficult to prove. In Fig. \ref{fig:e-tradeoff}, for two packets, we plot $U_1$, $U_2$ as a function of time and energy dedicated to the first packet $E_1$ and $t_1$, where $E_1+E_2=E$. Here again, we see that sending the shorter packet first is optimal in terms of minimizing the cost.

Finally, in Fig. \ref{fig:e-tradeoff}, for two packets, we plot the energy and time allotted to the two packets by the optimal algorithm, assuming bits of packet $1$ are sent before packet $2$'s. Blue curve is for $U_1$ and red for $U_2$, and where solid triangle and $\star$ and the value of $U_1$ evaluated at $(E_1, t_1)$ and $U_2$ at $E-E_1, t_2$, output by the algorithm. The surface curves in Fig. \ref{fig:e-tradeoff} also show the joint convexity of the considered cost function.
\begin{figure}[h]
\centering
\includegraphics[width=3in]{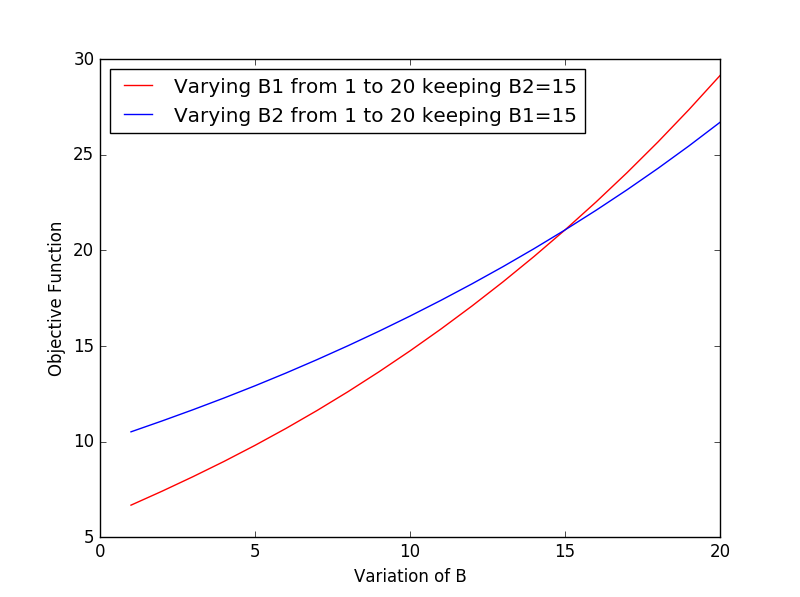}
\caption{Two packets with $B_1=15$ bits and $B_2 = 20$ bits with total energy $E=50$ Joules, where sending the shorter packet first is optimal for Problem \ref{eq:prob}.}
\label{fig:monotone}
\end{figure}
\begin{figure}[h]
\centering
\includegraphics[width=3in]{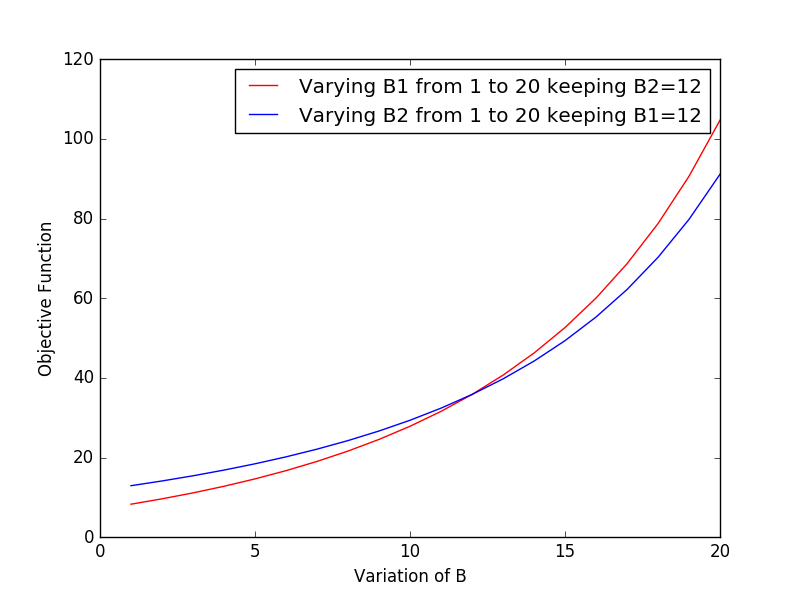}
\caption{Two packets with $B_1=12$ bits and $B_2 = 20$ bits with total energy $E=20$ Joules, where sending the shorter packet first is optimal for Problem \ref{eq:prob}.}
\label{fig:shortex}
\end{figure}

\begin{figure}[h]
\centering
\includegraphics[width=3.5in]{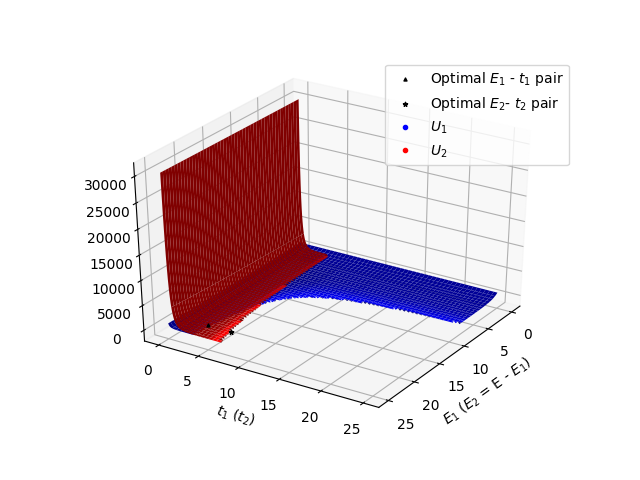}
\caption{Illustration of energy-time allocated to different packets for Problem \ref{eq:prob}.}
\label{fig:e-tradeoff}
\end{figure}

\section{Conclusions}
In this paper, we have introduced a somewhat unexplored problem of energy-distortion tradeoff under a delay cost, that is closely related to the rate-distortion problem with finite delays. 
This paper presents only preliminary and limited results on the considered problem. An important goal of this paper is to attract attention towards the considered problem, which we believe is not only practically relevant but also theoretically challenging, since it covers two fundamental and closely related important problems: rate-distortion problem with finite delays and generalized age of information problem.

\bibliographystyle{IEEEtran}
\bibliography{Research}

\appendices
\section{Proof of Theorem \ref{thm:jointconv}}\label{app:jointconv}
Here we give a proof about joint convexity of cost function $U$, where
$U=\sum_{i=1}^{n}2^{B_{i}-\hat{B_{i}}}+a_{i}\ell_{i}$, for some constant $a_i > 0$ that depends on the order of transmission of packets $\pi$.

\vspace{0.1in}

There are $2n$ unknown variables ($e_{i}s$ and $\ell_{i}s$ for each
packet). So Hessian matrix is $2n\times2n$ sized matrix given by

\vspace{0.1in}

$H = \begin{bmatrix}\text{\ensuremath{\frac{d^{2}U}{de_{1}^{2}}}} & \text{\ensuremath{\frac{d^{2}U}{d\ell_1de_{1}}}} & 0 & 0 & 0 & 0\\
\text{\ensuremath{\frac{d^{2}U}{de_{1}d\ell_1}}} & \frac{d^{2}U}{d\ell_1^{2}} & 0 & 0 & 0 & 0\\
0 & 0 & - & - & 0 & 0\\
0 & 0 & - & - & 0 & 0\\
0 & 0 & 0 & 0 & \text{\ensuremath{\frac{d^{2}U}{de_{n}^{2}}}} & \text{\ensuremath{\frac{d^{2}U}{d\ell_{n}de_{n}}}}\\
0 & 0 & 0 & 0 & \ensuremath{\frac{d^{2}U}{de_{n}d\ell_{n}}} & \frac{d^{2}U}{d\ell_{n}^{2}}
\end{bmatrix}$

\vspace{0.1in}

Clearly, the Hessian matrix is block diagonal, and since eigen-values
of block diagonal matrix are eigen-values of each block, to prove the
joint convexity, we need to prove that the eigen-values of each block
are positive, to prove that the Hessian matrix is positive semi-definite. 

Thus, it is sufficient to show that any of the block, say the first block 
$H_1=\begin{bmatrix}\ensuremath{\frac{d^{2}U}{de_{1}^{2}}} & \text{\ensuremath{\frac{d^{2}U}{d\ell_1de_{1}}}}\\
\text{\ensuremath{\frac{d^{2}U}{de_{1}d\ell_1}}} & \frac{d^{2}U}{d\ell_1^{2}},
\end{bmatrix}$
is positive-definite to prove the joint-convexity of $H$.

We use the Sylvester's
criterion for this purpose, that states that a $m\times m$ Hermitian matrix $S$ is positive-definite
if and only if all the upper left $k\times k$ corner of $S$ $\forall k, 1\le k \le n$ have a positive determinant:
Writing derivative terms
\vspace{0.1in}

$\frac{d^{2}U}{d e_1^2}=P\frac{(\ell_1+1)\ell_1}{(\ell_{1}+e_{1})^{2}}$,

\vspace{0.1in}

$\frac{d^{2}U}{de_{1}d\ell_1}=\frac{d^{2}U}{d\ell_1de_{1}}=-P\left(\frac{(\ell_1+1)e_{1}}{(\ell_1+e_{1})^{2}}+\frac{\ell_1}{\ell_1+e_{1}}\log_{2}(\frac{\ell_1}{\ell_1+e_{1}})\right),$ 

\vspace{0.1in}

$\frac{d^{2}U}{d\ell_1^{2}}=P\left(\left(\log_{2}(\frac{\ell_1}{\ell_1+e_{1}})+\frac{e_{1}}{e_{1}+\ell_1}\right)^{2}+\frac{e_{1}^{2}}{(\ell_1+e_{1})^{2}\ell_1}\right),$

where $P=2^{B_{1}}\left(\frac{e_{1}+\ell_1}{\ell_1}\right)^{-\ell_1}$.

Now using Sylvester's criteria, upper left $1\times1$ matrix is $[\ensuremath{\frac{d^{2}U}{de_{1}^{2}}}]$
which clearly has positive determinant. Only thing left in prove is
to show that $H_1$ has a positive determinant. $\det(H_1)$
\begin{align*} &=P^{2}\frac{(\ell_1+1)\ell_1}{(l_{1+}e_{1})^{2}}\left(\log_{2}^{2}\left(\frac{\ell_1}{\ell_1+e_{1}}\right)+2\log_{2}\left(\frac{\ell_1}{\ell_1+e_{1}}\right)\frac{e_{1}}{e_{1}+\ell_1}\right. \\
& \left. +\frac{e_{1}^{2}}{(\ell_1+e_{1})^{2}}+\frac{e_{1}^{2}}{(\ell_1+e_{1})^{2}\ell_1}\right)
\\ &-P^{2}\left(\frac{\ell_1^{2}}{(\ell_1+e_{1})^{2}}\log_{2}^{2}\left(\frac{\ell_1}{\ell_1+e_{1}}\right)\right. 
\\ &\left.+\frac{(\ell_1+1)^{2}e_{1}^{2}}{(\ell_1+e_{1})^{4}}+2\log_{2}\left(\frac{\ell_1}{\ell_1+e_{1}}\right)\frac{(\ell_1+1)\ell_1e_{1}}{(\ell_1+e_{1})^{3}}\right),\\
&= P^{2}\left(\frac{\ell_1}{(\ell_1+e_{1})^{2}}\log_{2}^{2}\left(\frac{\ell_1}{\ell_1+e_{1}}\right)\right),
\end{align*}
which is clearly positive, hence joint convexity of our objective
function is proved.


\end{document}

%% file: input.tex
\usepackage{amsfonts}
\usepackage{times}
\usepackage{latexsym}
\usepackage{amssymb}
\usepackage{amsmath}
\usepackage{cite}
\usepackage{verbatim}

\def\greedy{{\mathsf{GREEDY}}}

\def\bb0{{\mathbb{0}}}

\def\bb{{\mathbf{b}}}

\def\b0{{\mathbf{0}}}

\def\bS{{\mathbf{S}}}

\def\b1{{\mathbf{1}}}

\def\cP{\mathcal{P}}

\def\sf0{{\mathsf{0}}}
